\documentclass[conference,a4paper]{IEEEtran}

\newtheorem{thm}{Theorem}
\newtheorem{cor}{Corollary}
\newtheorem{lem}{Lemma}

\usepackage[final]{graphicx}
\usepackage[reqno]{amsmath}
\usepackage{amssymb}
\usepackage{subfig}
\usepackage{epstopdf}
\usepackage{bm}

\begin{document}

\sloppy

\title{Flexible Backhaul Design and Degrees of Freedom \\ for Linear Interference Networks}
\author{\IEEEauthorblockN{Aly El Gamal and Venugopal V.~Veeravalli}
 \IEEEauthorblockA{ECE Department and Coordinated Science Laboratory\\University of Illinois at Urbana-Champaign\\ Email: \{elgamal1,vvv\}@illinois.edu}}

\maketitle

\begin{abstract}
The considered problem is that of maximizing the degrees of freedom (DoF) in cellular downlink, under a \emph{backhaul load} constraint that limits the number of messages that can be delivered from a centralized controller to the base station transmitters. A linear interference channel model is considered, where each transmitter is connected to the receiver having the same index as well as one succeeding receiver. The backhaul load is defined as the sum of all the messages available at all the transmitters normalized by the number of users. When the backhaul load is constrained to an integer level $B$, the asymptotic per user DoF is shown to equal $\frac{4B-1}{4B}$, and it is shown that the optimal assignment of messages to transmitters is asymmetric and satisfies a local cooperation constraint and that the optimal coding scheme relies only on zero-forcing transmit beamforming. Finally, an extension of the presented coding scheme for the case where $B=1$ is shown to apply for more general locally connected and two-dimensional networks.
\end{abstract}

\section{Introduction}
Managing wireless interference through infrastructural enhancements is a major consideration for next generation cellular networks. One example of such an enhancement is in cellular downlink through the assignment of one receiver's message to multiple base station transmitters and managing interference through a Coordinated Multi-Point Transmission (CoMP) scheme. 
The cost of delivering messages to multiple transmitters over a backhaul link is highlighted in this work. 

In~\cite{ElGamal-Annapureddy-Veeravalli-arXiv12}, the degrees of freedom (DoF) gain offered by CoMP transmission in Wyner's linear interference networks~\cite{Wyner} was studied, under a cooperation constraint that limits the number of transmitters at which each message can be available by a number $M$. The asymptotic limit of the per user DoF as the number of users goes to infinity was shown to be $\frac{2M}{2M+1}$, and was shown to be achieved by a simple coding scheme that relies only on zero-forcing transmit beamforming. It is to be noted that the maximum transmit set size constraint of $M$ is not met tightly for all messages in the optimal message assignment scheme presented in~\cite{ElGamal-Annapureddy-Veeravalli-arXiv12}. In this work, we therefore consider a cooperation constraint that is more general and relevant to many scenarios of practical significance. In particular, we define the \emph{backhaul load} constraint $B$ as the ratio between the sum of the transmit set sizes for all the messages and the number of users. In other words, we allow the transmit set size constraints to vary across the messages, while maintaining a constraint on the average transmit set size of $B$.  We establish in this paper that the asymptotic per user DoF in this new setting is $\frac{4B-1}{4B}$, which is larger than the per user DoF of $\frac{2B}{2B+1}$ obtained with the more stringent per message transmit set size constraint of $B$.

Furthermore, we show that the scheme that achieves the optimal DoF of $\frac{4B-1}{4B}$ uses only zero-forcing beamforming at the transmitters, and assigns messages non-uniformly across the transmitters, with some messages being assigned to more than $B$ transmitters and others being assigned to fewer than $B$ transmitters. We show that these insights can apply to more general channel models than the simple linear model considered in this work.

We describe the system model in Section~\ref{sec:systemmodel}. We then provide an illustrative example for the considered problem in Section~\ref{sec:example}. The main result is proved in Section~\ref{sec:main}. We then discuss the result and its generalizations in Section~\ref{sec:discussion}. Finally, we provide concluding remarks in Section~\ref{sec:conclusion}.
\section{System Model and Notation}\label{sec:systemmodel}
We use the standard model for the $K-$user interference channel with single-antenna transmitters and receivers,
\begin{equation}
Y_i(t) = \sum_{j=1}^{K} H_{i,j}(t) X_j(t) + Z_i(t),
\end{equation}
where $t$ is the time index, $X_j(t)$ is the transmitted signal of transmitter $j$, $Y_i(t)$ is the received signal at receiver $i$, $Z_i(t)$ is the zero mean unit variance Gaussian noise at receiver $i$, and $H_{i,j}(t)$ is the channel coefficient from transmitter $j$ to receiver $i$ over the time slot $t$. We remove the time index in the rest of the paper for brevity unless it is needed. For any set ${\cal A} \subseteq [K]$, we use the abbreviations $X_{\cal A}$, $Y_{\cal A}$, and $Z_{\cal A}$ to denote the sets $\left\{X_i, i\in {\cal A}\right\}$, $\left\{Y_i, i\in {\cal A}\right\}$, and $\left\{Z_i, i\in {\cal A}\right\}$, respectively. Finally, we use $[K]$ to denote the set $\{1,2,\ldots,K\}$, and use $\phi$ to denote the empty set.

\subsection{Channel Model}
Each transmitter is connected to its corresponding receiver as well as one following receiver, and the last transmitter is only connected to its corresponding receiver. More precisely,

\begin{equation}\label{eq:channel}
H_{i,j} = 0 \text { iff } i \notin \{j,j+1\},\forall i,j \in [K],
\end{equation}
and all non-zero channel coefficients are drawn independently from a continuous joint distribution. Finally, we assume that global channel state information is available at all transmitters and receivers. The channel model is illustrated for $K=3$ in Figure~\ref{fig:wynermodel}.

\begin{figure}[htb]
\centering
\includegraphics[width=0.8\columnwidth]{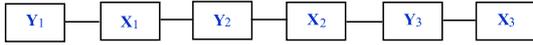}
\caption{Wyner's linear asymmetric model for $K=3$. In the figure, a solid line connects a transmitter-receiver pair if and only if the channel coefficient between them is non-zero.}
\label{fig:wynermodel}
\end{figure}

\subsection{Message Assignment}
For each $i \in [K]$, let $W_i$ be the message intended for receiver $i$, and ${\cal T}_i \subseteq [K]$ be the transmit set of receiver $i$, i.e., those transmitters with the knowledge of $W_i$. The transmitters in ${\cal T}_i$ cooperatively transmit the message $W_i$ to the receiver $i$. The average transmit set size is upper bounded by an integer valued backhaul load constraint $B$,
\begin{equation}\label{eq:backhaul_constraint}
\frac{\sum_{i=1}^K |{\cal T}_i|}{K} \leq B.
\end{equation}

\subsection{Message Assignment Strategy}

A message assignment strategy is defined by a sequence of supersets. The $k^{th}$ element in the sequence consists of the transmit sets for a $k-$user channel. We use message assignment strategies to define a pattern for assigning messages to transmitters in large networks. 

\subsection{Local Cooperation}\label{sec:localcooperation}

We say that a message assignment strategy satisfies the local cooperation constraint, if and only if there exists a function $r(K)$ such that $r(K)=o(K)$, and for every $K\in{\bm Z}^{+}$, the transmit sets defined by the strategy for a $K-$user channel satisfies the following,
\begin{equation}\label{eq:localcooperation}
{\cal T}_{i} \subseteq \{i-r(K),i-r(K)+1,\ldots,i+r(K)\}, \forall i\in[K].
\end{equation}

\subsection{Degrees of Freedom}
Let $P$ be the average transmit power constraint at each transmitter, and let ${\cal W}_i$ denote the alphabet for message $W_i$. Then the rates $R_i(P) = \frac{\log|{\cal W}_i|}{n}$ are achievable if the decoding error probabilities of all messages can be simultaneously made arbitrarily small for a large enough coding block length $n$, and this holds for almost all channel realizations. The degrees of freedom $d_i, i\in[K],$ are defined as $d_i=\lim_{P \rightarrow \infty} \frac{R_i(P)}{\log P}$. The DoF region ${\cal D}$ is the closure of the set of all achievable DoF tuples. The total number of degrees of freedom ($\eta$) is the maximum value of the sum of the achievable degrees of freedom, $\eta=\max_{\cal D} \sum_{i \in [K]} d_i$.

For a $K$-user channel, we define $\eta(K,B)$ as the best achievable $\eta$ over all choices of transmit sets satisfying the backhaul load constraint in \eqref{eq:backhaul_constraint}. 
In order to simplify our analysis, we define the asymptotic per user DoF $\tau(B)$ to measure how $\eta(K,B)$ scales with $K$ while all other parameters are fixed,
\begin{equation}
\tau(B) = \lim_{K\rightarrow \infty} \frac{\eta(K,B)}{K},
\end{equation}

We call a message assignment strategy \emph{optimal} for a sequence of $K-$user channels, $K\in\{1,2,\ldots\}$,  if and only if there exists a sequence of coding schemes achieving $\tau(B)$ using the transmit sets defined by the message assignment strategy.

\section{Example: $B=1$}\label{sec:example}
Before introducing the main result, we illustrate through a simple example that the potential flexibility in the backhaul design according to the constraint in~\eqref{eq:backhaul_constraint} can offer DoF gains over a traditional design where all messages are assigned to the same number of transmitters. We know from~\cite{ElGamal-Annapureddy-Veeravalli-arXiv12} that any asymptotic per user DoF greater than $\frac{2}{3}$ cannot be achieved through assigning each message to one transmitter. We now show that $\tau(B=1)\geq\frac{3}{4}$, by allowing few messages to be available at more than one transmitter at the cost of not transmitting other messages. Consider the following assignment of the first four messages, ${\cal T}_1=\{1,2\}$, ${\cal T}_2=\{2\}$, ${\cal T}_3=\phi$, and ${\cal T}_4=\{3\}$. Message $W_1$ is transmitted through $X_1$ to $Y_1$ without interference. Since  the channel state information is known at the second transmitter, the transmit beam for $W_1$ at $X_2$ can be designed to cancel the interference caused by $W_1$ at $Y_2$, and then $W_2$ can be transmitted through $X_2$ to $Y_2$ without interference. Finally, $W_4$ is transmitted through $X_3$ to $Y_4$ without interference. It follows that the sum DoF for the first four messages $\sum_{i=1}^4 d_i \geq 3$. Since the fourth transmitter is inactive, the subnetwork consisting of the first four users does not interfere with the rest of the network, and hence, we can see that $\tau(B=1) \geq \frac{3}{4}$ through similar assignment of messages in each consecutive $4$-user subnetwork.

\section{Main Result}\label{sec:main}
We now characterize the asymptotic per user DoF $\tau(B)$ for any integer value of the backhaul load constraint.
\begin{thm}\label{thm:main_result}
The asymptotic per user DoF $\tau(B)$ is given by,
\begin{equation}\label{eq:main_result}
\tau(B)=\frac{4B-1}{4B}, \forall B\in{\bf Z}^+.
\end{equation}
\end{thm}
\begin{proof}
We provide the proof for the inner and outer bounds in Section~\ref{sec:ib} and Section~\ref{sec:ub}, respectively.
\end{proof}
\subsection{Coding Scheme}\label{sec:ib}
We treat the network as a set of subnetworks, each consisting of consecutive $4B$ transceivers. The last transmitter of each subnetwork is deactivated to eliminate {\em inter-subnetwork} interference. It then suffices to show that $4B-1$ DoF can be achieved in each subnetwork. Without loss of generality, consider the cluster of users with indices in the set $[4B]$. We define the following subsets of $[4B]$,
\begin{eqnarray*}
{\cal S}_1 &=& [2B]
\\{\cal S}_2 &=& \{2B+2,2B+3,\ldots,4B\}
\end{eqnarray*}
We next show that each user in ${\cal S}_1 \cup {\cal S}_2$ achieves one degree of freedom, while message $W_{2B+1}$ is not transmitted. Let the message assignments be as follows,\\

${\cal T}_{i}=
\begin{cases}
\{i,i+1,\ldots,2B\}, \quad &\forall i \in {\cal S}_1,\\
\{i-1,i-2,\ldots,2B+1\},\quad &\forall  i \in {\cal S}_2,
\end{cases}$\\
and note that $\frac{\sum_{i=1}^{4B} |{\cal T}_i|}{4B}=B$, and hence, the constraint in~\eqref{eq:backhaul_constraint} is satisfied. Now, due to the availability of channel state information at the transmitters, the transmit beams for message $W_i$ can be designed to cancel its effect at receivers with indices in the set ${\cal C}_i$, where,\\

${\cal C}_{i}=
\begin{cases}
\{i+1, i+2, \ldots,2B\},\quad  &\forall i \in {\cal S}_1\\
\{i-1, i-2, \ldots,2B+2\},\quad &\forall  i \in {\cal S}_2
\end{cases}$\\

Note that both ${\cal C}_{2B}$ and ${\cal C}_{2B+2}$ equal the empty set, as both $W_{2B}$ and $W_{2B+2}$ do not contribute to interfering signals at receivers in the set $Y_{{\cal S}_1} \cup Y_{{\cal S}_2}$. The above scheme for $B=2$ is illustrated in Figure~\ref{fig:btwo}.
We conclude that each receiver whose index is in the set ${\cal S}_1\cup{\cal S}_2$ suffers only from Gaussian noise, thereby enjoying one degree of freedom. Since $|{\cal S}_1\cup{\cal S}_2|=4B-1$, it follows that $\sum_{i=1}^{4B} d_i \geq 4B-1$. Using a similar argument for each following subnetwork, we establish that $\tau(B) \geq \frac{4B-1}{4B}$, thereby proving the lower bound of Theorem~\ref{thm:main_result}.

\begin{figure}[htb]
\centering
\includegraphics[width=0.5\columnwidth]{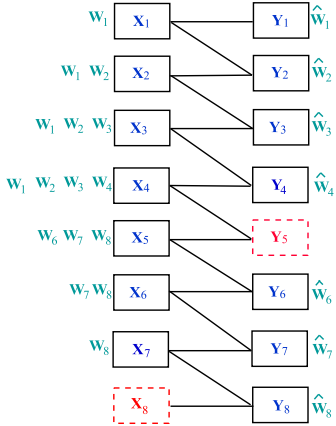}
\caption{Achieving $7/8$ per user DoF with a backhaul constraint $B=2$. The figure shows only signals corresponding to the first subnetwork in a general $K-$user network. The signals in the dashed boxes are deactivated. Note that the deactivation of $X_8$ splits this part of the network from the rest.}
\label{fig:btwo}
\end{figure} 

We note that the illustrated message assignment strategy satisfies the local cooperation constraint of~\eqref{eq:localcooperation}. In other words, the network can be split into subnetworks, each of size $4B$, and the messages corresponding to users in a subnetwork can only be assigned to transmitters with indices in the same subnetwork.
\subsection{Upper Bound}\label{sec:ub}
We prove the converse of Theorem~\ref{thm:main_result} in two steps. First, we provide an information theoretic argument in Lemma~\ref{lem:aux_ub} to prove an upper bound on the DoF of any network that has a subset of messages whose transmit set sizes are bounded. We then finalize the proof with a combinatorial argument that shows the existence of such a subset of messages in any assignment of messages satisfying the backhaul constraint of~\eqref{eq:backhaul_constraint}.

In order to prove the information theoretic argument in Lemma~\ref{lem:aux_ub}, we use Lemma $4$ from~\cite{ElGamal-Annapureddy-Veeravalli-arXiv12}, which we restate below. For any set of receiver indices ${\cal A} \subseteq [K]$, define $U_{\cal A}$ as the set of indices of transmitters that exclusively carry the messages for the receivers in ${\cal A}$, and its complement $\bar{U}_{\cal A}$. More precisely, $\bar{U}_{\cal A} = \cup_{i \notin {\cal A}} {\cal T}_i$.
\begin{lem} [\cite{ElGamal-Annapureddy-Veeravalli-arXiv12}] \label{lem:dofouterbound}
If there exists a set ${\cal A}\subseteq [K]$, a function $f_1$, and a function $f_2$ whose definition does not depend on the transmit power constraint $P$, and $f_1\left(Y_{\cal A},X_{U_{\cal A}}\right)=X_{\bar{U}_{\cal A}}+f_2(Z_{\cal A})$, then the sum DoF $\eta \leq |{\cal A}|$.
\end{lem}

We also need~\cite[Corollary $3$]{ElGamal-Annapureddy-Veeravalli-arXiv12} in the proof of Lemma~\ref{lem:aux_ub}; we restate it for the considered system model.
\begin{cor}[\cite{ElGamal-Annapureddy-Veeravalli-arXiv12}]
\label{cor:dofouterbound}
For any $K-$user linear interference channel, if the size of the transmit set $|{\cal T}_i| \leq M, i\in[K],$ then any element $k \in {\cal T}_i$ such that $k \notin \{i-M,i-M+1,\ldots,i+M-1\}$ can be removed from ${\cal T}_i$, without decreasing the sum rate.
\end{cor}

We now make the following definition to use in the proof of the following lemma. For any set ${\cal S}\subseteq [K]$, let $g_{\cal S}: {\cal S} \rightarrow \{1,2,\ldots,|{\cal S}|\}$ be a function that returns the ascending order of any element in the set ${\cal S}$, e.g., $g_{\cal S}\left(\min \left\{i: i\in{\cal S}\right\}\right)=1$ and $g_{\cal S}\left(\max \left\{i: i\in{\cal S}\right\}\right)=|{\cal S}|$

\begin{lem}\label{lem:aux_ub}
For any $K-$user linear interference channel with DoF $\eta$, if there exists a subset of messages ${\cal S} \subseteq [K]$ such that each message in ${\cal S}$ is available at a maximum of $M$ transmitters, i.e., $|{\cal T}_i| \leq M, \forall i \in {\cal S}$, then the DoF is bounded by,
\begin{equation}
\eta \leq K-\frac{|{\cal S}|}{2M+1}+C_K,
\end{equation}
where $\lim_{K \rightarrow \infty} \frac{C_K}{K} = 0$.
\end{lem}
\begin{proof}
We use Lemma~\ref{lem:dofouterbound} with a set ${\cal A}$ such that the size of the complement set $|\bar{\cal A}|=\frac{|{\cal S}|}{2M+1}-o(K)$. We define the set ${\cal A}$ such that ${\bar{\cal A}}=\{i: i\in{\cal S}, g_{\cal S}(i)= (2M+1)(j-1)+M+1, j \in {\bm{Z}^+}\}$. 

Now, we let $s_1$, $s_2$ be the smallest two indices in $\bar{\cal A}$. We see that $g_{\cal S}(s_1)=M+1, g_{\cal S}(s_2)=3M+2$. Note that $X_1+\frac{Z_1}{H_{1,1}}=\frac{Y_1}{H_{1,1}}$, and
\begin{equation*}
X_2+\frac{Z_2-\frac{H_{2,1}}{H_{1,1}}Z_1}{H_{2,2}}=\frac{Y_2-\frac{H_{2,1}}{H_{1,1}}Y_1}{H_{2,2}}.
\end{equation*}
Similarly, it is clear how the first $s_1-1$ transmit signals $X_{[s_1-1]}$ can be recovered from the received signals $Y_{[s_1-1]}$ and linear combinations of the noise signals $Z_{[s_1-1]}$. In what follows, we show how to reconstruct a noisy version of the signals $\left\{X_{s_1},X_{s_1+1},\ldots,X_{s_2-1}\right\}$, where the reconstruction noise is a linear combination of the signals $Z_{\cal A}$. Then it will be clear by symmetry how the remaining transmit signals can be reconstructed. 

We now notice that it follows from Corollary~\ref{cor:dofouterbound} that message $W_{s_1}$ can be removed from any transmitter in ${\cal T}_{s_1}$ whose index is greater than $s_1+M-1$, without affecting the sum rate. Similarly, there is no loss in generality in assuming that $\forall s_i\in{\cal S}, s_i \neq s_1$, ${\cal T}_{s_i}$ does not have an element with index less than $s_i-M$. Since $s_i-s_1\geq g_{\cal S}(s_i)-g_{\cal S}(s_1)\geq 2M+1$, it follows that $X_{s_1+M} \in X_{U_{\cal A}}$. The signal $X_{s_1+M+1}+\frac{Z_{s_1+M+1}}{H_{s_1+M+1,s_1+M+1}}$ can be reconstructed from $Y_{s_1+M+1}$ and $X_{s_1+M}$. Then, it can be seen that the transmit signals $\left\{X_{s_1+M+2},X_{s_1+M+3},\ldots,X_{s_2-1}\right\}$ can be reconstructed from $\left\{Y_{s_1+M+1}, Y_{s_1+M+2},\ldots,Y_{s_2-1}\right\}$, and linear combinations of the noise signals $\left\{Z_{s_1+M+1},Z_{s_1+M+2},\ldots,Z_{s_2-1}\right\}$. Similarly, since $X_{s_1+M}$ is known, the transmit signals $\left\{X_{s_1+M-1},X_{s_1+M-2},\ldots,X_{s_1}\right\}$ can be reconstructed from $\left\{Y_{s_1+M}, Y_{s_1+M-1},\ldots,Y_{s_1+1}\right\}$, and linear combinations of the noise signals $\left\{Z_{s_1+M},Z_{s_1+M-1},\ldots,Z_{s_1+1}\right\}$. By following a similar argument to reconstruct all transmit signals from the signals $Y_{\cal A}$, $X_{U_{\cal A}}$, and linear combinations of the noise signals $Z_{\cal A}$, we can show the existence of functions $f_1$ and $f_2$ of Lemma~\ref{lem:dofouterbound} to complete the proof.
\end{proof}

We now explain how Lemma~\ref{lem:aux_ub} can be used to prove that $\tau(B=1) \leq \frac{3}{4}$. For any message assignment satisfying~\eqref{eq:backhaul_constraint} for a $K-$user channel, let $R_{j}$ be defined as follows for every $j\in\{0,1,\ldots,K\}$,
\begin{equation}\label{eq:concentration}
R_{j} = \frac{|\left\{i:i\in[K], |{\cal T}_i|=j\right\}|}{K}.
\end{equation}
$R_{j}$ is the fraction of users whose messages are available at exactly $j$ transmitters. Now, if $R_0+R_1 \geq \frac{3}{4}$, then Lemma~\ref{lem:aux_ub} can be used directly to show that $\eta \leq \frac{3K}{4}+o(K)$. Otherwise, more than $\frac{K}{4}$ users have their messages at two or more transmitters, and it follows from~\eqref{eq:backhaul_constraint} that $R_0 \geq \sum_{j=2}^K R_j \geq \frac{1}{4}$, and hence, $\eta \leq (1-R_0)K \leq \frac{3K}{4}$.

We generalize the above argument in the proof of the following lemma to complete the proof that $\tau(B)\leq\frac{4B-1}{4B},\forall B\in \bm{Z}^+$.
\begin{lem}\label{lem:comb_argument}
For any message assignment satisfying~\eqref{eq:backhaul_constraint} for a $K-$user channel with an average transmit set size constraint $B$, there exists an integer $M\in\{0,1,\ldots,K\}$, and a subset ${\cal S} \subseteq [K]$ whose size $|{\cal S}|\geq \frac{2M+1}{4B} K$, such that each message in ${\cal S}$ is available at a maximum of $M$ transmitters, i.e., $|{\cal T}_i| \leq M, \forall i\in{\cal S}$.
\end{lem}
\begin{proof}
Fix any message assignment satisfying~\eqref{eq:backhaul_constraint} for a $K-$user channel with backhaul constraint $B$, and let $R_j,j\in\{0,1,\ldots,K\}$ be defined as in~\eqref{eq:concentration}. If $\sum_{j=2B}^{K} R_j \leq \frac{1}{4B}$, then more than $\frac{4B-1}{4B}K$ users have a transmit set whose size is at most $2B-1$, and the lemma follows  with $M=2B-1$. It then suffices to assume that $\sum_{j=2B}^{K} R_j > \frac{1}{4B}$ in the rest of the proof. We show in the following that there exists an integer $M \in\{0,\ldots,2B-2\}$ such that $\sum_{j=0}^M R_j > \frac{2M+1}{4B}$, thereby completing the proof of the lemma.

Define $R_j^*, j\in\{0,1,\ldots,2B\}$ such that $R_0^*=R_{2B}^*=\frac{1}{4B}$, and $R_j^*=\frac{1}{2B}, \forall j\in\{1,\ldots,2B-1\}$. Now, note that $\sum_{j=0}^{2B} R_j^*=1$, and $\sum_{j=0}^{2B} jR_j^*=B$. It follows that if $R_j=R_j^*, \forall j\in\{0,\ldots,2B\}$, and $R_j=0, \forall j \geq 2B+1$, then the constraint in~\eqref{eq:backhaul_constraint} is tightly met, i.e., $\frac{\sum_{i=1}^K |{\cal T}_i|}{K}=B$. We will use this fact in the rest of the proof. 

We prove the statement by contradiction. Assume that $\sum_{j=2B}^K R_j > R_{2B}^* = \frac{1}{4B}$, and that $\forall M\in\{0,1,\ldots,2B-2\}, \sum_{j=0}^M R_j \leq \sum_{j=0}^M R_{j}^* = \frac{2M+1}{4B}$. We know from~\eqref{eq:backhaul_constraint} that $\sum_{j=0}^{K} jR_j \leq \sum_{j=0}^{2B} jR_j^*=B$. Also, since $\sum_{j=0}^{K} R_j = \sum_{j=0}^{2B} R_j^*=1$ and $\sum_{j=2B}^K R_j > R_{2B}^*$, it follows that there exists an integer $M \in \{0,1,\ldots,2B-1\}$ such that $R_M > R_M^*$; let $m$ be the smallest such integer. Since $\sum_{j=0}^m R_j \leq \sum_{j=0}^m R_j^*$, and $\forall j\in\{0,1,\ldots,m-1\}, R_j \leq R_j^*$, we can construct another message assignment by removing elements from some transmit sets whose size is $m$, such that the new assignment satisfies~\eqref{eq:backhaul_constraint}, and has transmit sets ${\cal T}_i^*$ where $\forall j\in\{0,1,\ldots,m\}, |\{i:i\in[K], |{\cal T}_i^*| = j\}| \leq R_j^*$. By successive application of the above argument, we can construct a message assignment that satisfies~\eqref{eq:backhaul_constraint}, and has transmit sets ${\cal T}_i^*$ where $\forall j\in\{0,1,\ldots,2B-1\}, |\{i:i\in[K], |{\cal T}_i^*| = j\}| \leq R_j^*$ and $|\{i:i\in[K], |{\cal T}_i^*| \geq 2B\}| \geq R_{2B}^*$. Note that the new assignment has to violate~\eqref{eq:backhaul_constraint} since $\sum_{j=0}^{2B} jR_j^*=B$, and we reach a contradiction.
\end{proof}

We now know from lemmas~\ref{lem:aux_ub} and~\ref{lem:comb_argument} that under the backhaul load constraint of~\eqref{eq:backhaul_constraint}, the DoF for any $K-$user channel is upper bounded by $\frac{4B-1}{4B}K+o(K)$. It follows that the asymptotic per user DoF $\tau(B) \leq \frac{4B-1}{4B}$, thereby proving the upper bound of Theorem~\ref{thm:main_result}.

\section{Discussion and Generalizations}\label{sec:discussion}
\subsection{Maximum Transmit Set Size Constraint}
In~\cite{ElGamal-Annapureddy-Veeravalli-arXiv12}, we considered the problem where each transmit set size is bounded by a cooperation constraint $M$, i.e., $|{\cal T}_i| \leq M, \forall i\in[K]$. The DoF achieving coding scheme was then characterized for every value of $M$. We note that in the considered problem with an average transmit set size constraint $B$, the per user DoF $\tau(B)$ can be achieved using a combination of the schemes that are characterized as optimal in~\cite{ElGamal-Annapureddy-Veeravalli-arXiv12} for the cases of $M=2B-1$ and $M=2B$. We note that even though the maximum transmit set size constraint may not reflect a physical constraint, the solutions in~\cite{ElGamal-Annapureddy-Veeravalli-arXiv12} provide a useful toolset that can be used to achieve the optimal per user DoF value under the more natural constraint on the total backhaul load that is considered in this work.

\subsection{Locally Connected Networks}
 Using a convex combination of the schemes that are optimal under the maximum transmit set size constraint can also provide good coding schemes for the more general locally connected channel model that is considered in~\cite{ElGamal-Annapureddy-Veeravalli-arXiv12}, where each receiver can see interference from $L$ neighbouring transmitters. More precisely, for the following channel model,
\begin{eqnarray}\label{eq:general_channel}
&&H_{i,j} \text{ is not identically } 0, \nonumber\\&&\text { if and only if } i \in \left[j- \left \lfloor \frac{L}{2} \right \rfloor , j+ \left \lceil \frac{L}{2} \right \rceil \right].
\end{eqnarray}
 Let $\tau_L(B)$ be the asymptotic per user DoF for a locally connected channel defined in~\eqref{eq:general_channel} with connectivity parameter $L$. Then we can use a convex combination of the schemes that are characterized as optimal in~\cite{ElGamal-Annapureddy-Veeravalli-arXiv12} to achieve the inner bounds stated in Table~\ref{tab:results} for the case where $B=1$.

\begin{table}[b]
\begin{center}
    \begin{tabular}{ | l | l | l | l | l | l |}
    \hline
    & $L=2$ & $L=3$ & $L=4$ & $L=5$ & $L=6$ \\ \hline
     $\tau_L(B=1) \geq$ & $\frac{2}{3}$ & $ \frac{3}{5}$ & $\frac{5}{9}$ & $\frac{11}{21}$ & $\frac{1}{2}$ \\ \hline
    \end{tabular}
\end{center}
\caption{Achievable per user DoF values for locally connected channels with a backhaul constraint $\sum_{i=1}^{K} |{\cal T}_i| \leq K$.}
\label{tab:results}
\end{table}

Now, we note that the inner bounds stated in Table~\ref{tab:results} can be achieved through the use of only zero-forcing transmit beamforming. In other words, there is no need for the symbol extension idea required by the asymptotic interference alignment scheme of~\cite{Cadambe-IA}. In~\cite[Theorem $8$]{ElGamal-Annapureddy-Veeravalli-arXiv12}, it was shown that for $L \geq 2$, by allowing each message to be available at one transmitter, the asymptotic per user DoF is $\frac{1}{2}$; it was also shown in~\cite[Theorem $6$]{ElGamal-Annapureddy-Veeravalli-arXiv12} that the $\frac{1}{2}$ per user DoF value cannot be achieved through zero-forcing transmit forming for $L \geq 3$. In contrast, in Table~\ref{tab:results} it can be seen that for $L \leq 6$, the $\frac{1}{2}$ per user DoF value can be achieved through zero-forcing transmit beamforming and a flexible design of the backhaul links, without incurring additional overall load on the backhaul ($B=1$).

\subsection{Two-Dimensional Networks}
The insights we have in this work on the backhaul design for linear interference networks, may apply in denser networks by treating the denser network as a set of interfering linear networks. For example, consider the two-dimensional network depicted in Figure~\ref{fig:two-dim} where each transmitter is connected to four cell edge receivers. The precise channel model for a $K-$user channel is as follows,

\begin{eqnarray}\label{eq:twodim_channel}
&&H_{i,j} \text{ is not identically } 0, \text { if and only if }\nonumber\\&& i \in \left\{j,j+1,j+\left \lfloor \sqrt{K} \right \rfloor, j+\left \lfloor \sqrt{K} \right \rfloor+1\right\}.\nonumber\\
\end{eqnarray}

\begin{figure}[b]
  \centering
\subfloat[]{\label{fig:two-dim}\includegraphics[height=0.15\textwidth]{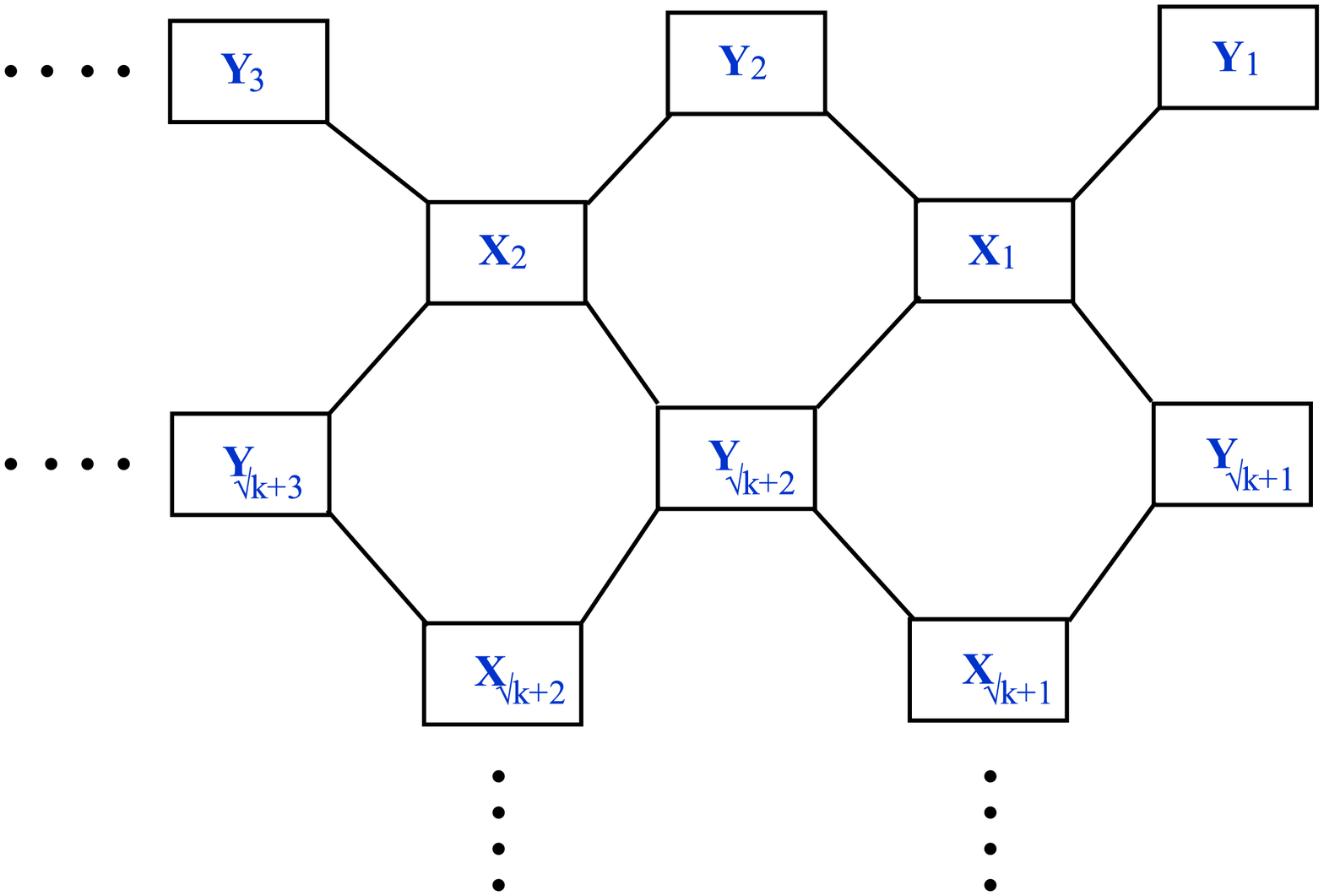}}                
\quad\quad\subfloat[]{\label{fig:twodim-codingscheme}\includegraphics[width=0.19\textwidth]{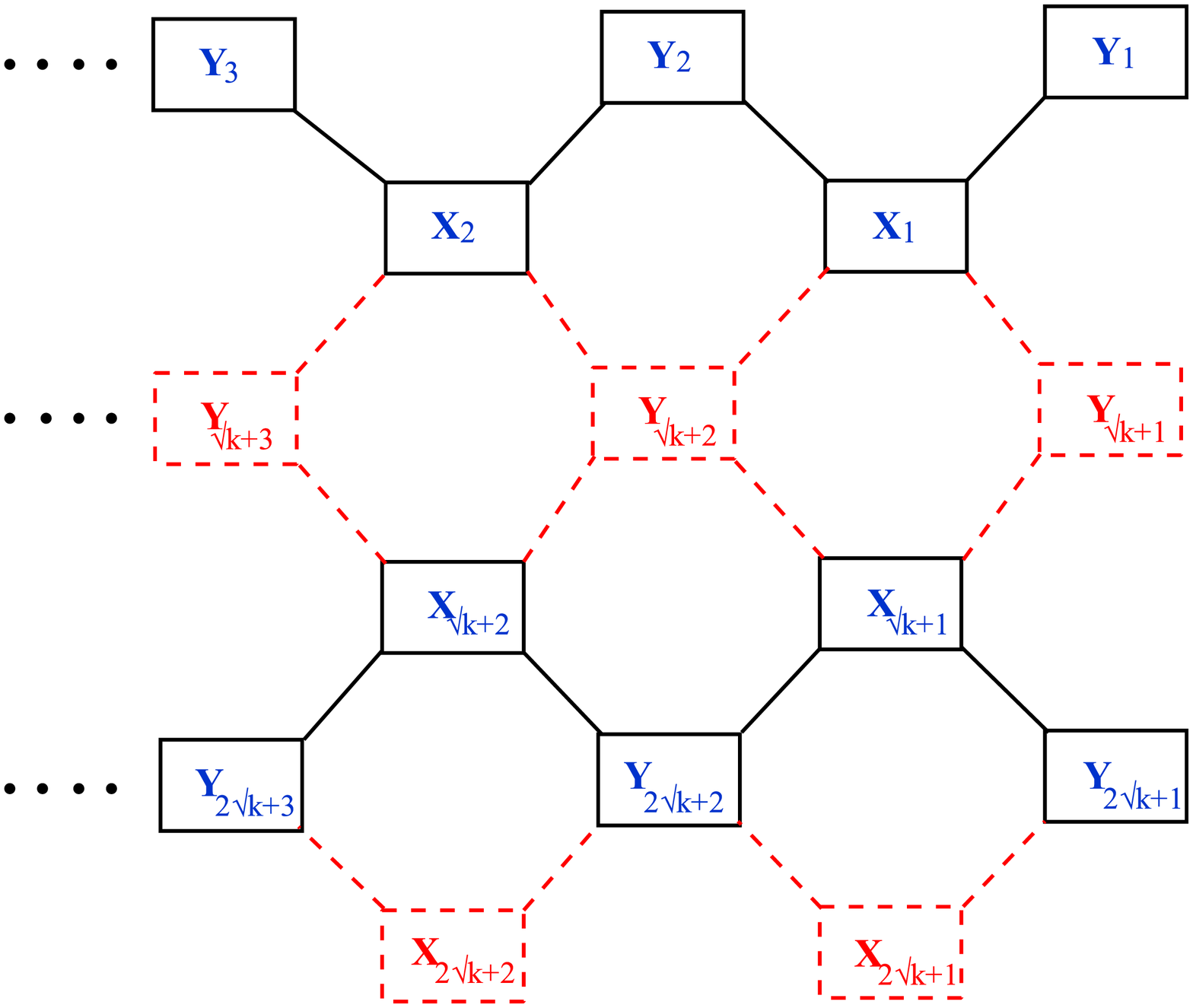}}
  \caption{Two dimensional interference network. In $(a)$, we plot the channel model, with each transmitter being connected to four surrounding cell edge receivers. In $(b)$, we show an example coding scheme where dashed red boxes and lines represent inactive nodes and edges. The signals $\{X_1,\ldots,X_{\sqrt{K}}\}$ and $\{Y_1,\ldots,Y_{\sqrt{K}}\}$ form a linear subnetwork. Similarly, the signals $\{X_{\sqrt{K}+1},\ldots,X_{2\sqrt{K}}\}$ and $\{Y_{2\sqrt{K}+1},\ldots,Y_{3\sqrt{K}}\}$ form a linear subnetwork}.
  \label{fig:two-dimensional}
\end{figure}

For this channel model, we can show that by assigning each message to one transmitter, i.e., imposing the constraint $|{\cal T}_i| \leq 1, \forall i\in[K],$ the asymptotic per user DoF is at most $\frac{1}{2}$, and the use of only zero-forcing transmit beamforming can lead to at most $\frac{4}{9}$ per user DoF. However, under the backhaul load constraint $\frac{\sum_{i=1}^K |{\cal T}_i|}{K}\leq 1$, a per user DoF value of $\frac{5}{9}$ can be achieved using only zero-forcing transmit beamforming. This can be done by deactivating every third row of transmitters, and splitting the rest of the network into non-interfering linear subnetworks (see Figure~\ref{fig:twodim-codingscheme}). 
In each subnetwork, a backhaul load constraint of $\frac{3}{2}$ is imposed. For example, the following constraint is imposed on the first row of users, $\frac{\sum_{i=1}^{\left \lfloor \sqrt{K} \right \rfloor} |{\cal T}_i|}{\left \lfloor \sqrt{K} \right \rfloor}\leq\frac{3}{2}$. A convex combination of the schemes that are characterized as optimal in~\cite{ElGamal-Annapureddy-Veeravalli-arXiv12} for the cases of maximum transmit set size constraints $M=2$ and $M=3$ is then used to achieve $\frac{5}{6}$ per user DoF in each active subnetwork while satisfying a backhaul load constraint of $\frac{3}{2}$. Since $\frac{2}{3}$ of the subnetworks are active, a per user DoF of $\frac{5}{9}$ is achieved while satisfying a backhaul load constraint of unity.

\section{Conclusion}\label{sec:conclusion}
We studied the potential gains offered by CoMP transmission in linear interference networks, through a backhaul load constraint that limits the average transmit set size across the users. We characterized the asymptotic per user DoF, and showed that the optimal coding scheme relies only on zero-forcing transmit beamforming. The backhaul constraint is satisfied in the optimal scheme by assigning some messages to more than $B$ transmitters and others to fewer than $B$ transmitters, where $B$ is the average transmit set size. We showed that local cooperation is sufficient to achieve the DoF in large linear interference networks. We also noted that the characterized asymptotic per user DoF for linear interference networks can be achieved by using a convex combination of the coding schemes that are identified as optimal in~\cite{ElGamal-Annapureddy-Veeravalli-arXiv12} under a cooperation constraint that limits the maximum size of a transmit set, as opposed to the average as we considered in this work. We then illustrated that these results hold in more general networks of practical relevance to achieve rate gains and simplify existing coding schemes. 
\bibliographystyle{IEEEtran}

\end{document}